\documentclass[a4paper]{article}
\usepackage{a4wide}
\usepackage{amsmath}
\usepackage{amsthm}
\usepackage{amssymb}
\usepackage{environ}
\usepackage{xspace}
\usepackage{stmaryrd}
\usepackage{mathrsfs}
\usepackage{url}
\usepackage{hyperref}
\usepackage{pifont,xcolor}

\usepackage{ifpdf}
\ifpdf
\usepackage{tikz}
\usetikzlibrary{arrows,automata}
\usetikzlibrary{arrows,automata,positioning,calc}
\usetikzlibrary{shapes,snakes}
\fi

\newcommand{\bsymb}{
{\begin{tikzpicture}
\filldraw[fill=white,draw=black] circle (2.8pt);
\end{tikzpicture}}
}
\newcommand{\hatbsymb}{
{\begin{tikzpicture}
\filldraw[fill=gray!90,draw=black] circle (2.8pt);
\end{tikzpicture}}
}
\newcommand{\asymb}{
{\begin{tikzpicture}
\filldraw[fill=white,draw=black] rectangle (5pt,5pt);
\end{tikzpicture}}
}
\newcommand{\hatasymb}{
{\begin{tikzpicture}
\filldraw[fill=gray!90,draw=black] rectangle (5pt,5pt);
\end{tikzpicture}}
}

\makeatletter
\newcommand{\subalign}[1]{%
  \vcenter{%
    \Let@ \restore@math@cr \default@tag
    \baselineskip\fontdimen10 \scriptfont\tw@
    \advance\baselineskip\fontdimen12 \scriptfont\tw@
    \lineskip\thr@@\fontdimen8 \scriptfont\thr@@
    \lineskiplimit\lineskip
    \ialign{\hfil$\m@th\scriptstyle##$&$\m@th\scriptstyle{}##$\hfil\crcr
      #1\crcr
    }%
  }%
}
\makeatother

\title{Erratum to ``Frequency Linear-time Temporal Logic''}

\author{
Benedikt Bollig, Normann Decker, and Martin Leucker
}
\date{}

      \theoremstyle{plain}
      \newtheorem{theorem}{Theorem}
      \newtheorem{lemma}{Lemma}

      \newtheorem{claim}{Claim}
      
      \theoremstyle{definition}

      \theoremstyle{remark}
      \newtheorem{remark}{Remark}

\renewcommand{\epsilon}{\varepsilon}

\colorlet{nred}{black}
\colorlet{nblue}{black}
\colorlet{ngreen}{black}

\newcommand{\Rat}{\mathbb{Q}}

\newcommand{\fltl}{\textrm{\textit{f}LTL}\xspace}

\newcommand{\U}{\ensuremath{\operatorname{\mathsf{U}}}\xspace}
\newcommand{\myU}[1]{\mathrel{\mathsf{U}}^{\ge #1}}

\newcommand{\X}{\ensuremath{\operatorname{\mathsf{X}}}\xspace}
\newcommand{\G}{\ensuremath{\operatorname{\mathsf{G}}}\xspace}

\newcommand{\F}{\ensuremath{\operatorname{\mathsf{F}}}\xspace}

\newcommand{\true}{\top}
\newcommand{\false}{\bot}

\newcommand{\op}{\alpha}

\newcommand{\IS}{\mathsf{OP}}

\renewcommand{\phi}{\varphi}
\newcommand{\N}{\mathbb{N}}
\newcommand{\occnumb}[2]{|\!|#1,#2|\!|}

\newcommand{\inc}{\ensuremath{\mathsf{inc}}\xspace}
\newcommand{\testz}{\ensuremath{\mathsf{zero}}\xspace}
\newcommand{\dec}{\ensuremath{\mathsf{dec}}\xspace}

\newcommand{\sufx}[2]{{#1}|^{#2}}
\newcommand{\loc}{\ell}
\newcommand{\Loc}{\mathit{Loc}}
\newcommand{\Trans}{T}

\newcommand{\src}{\mathit{src}}
\newcommand{\trg}{\mathit{trg}}
\newcommand{\opt}{\mathit{op}}

\newcommand{\valone}{m}
\newcommand{\valtwo}{n}
\newcommand{\ttrans}[1]{\xrightarrow{#1}}

\newcommand{\enc}[1]{\mathsf{enc}(#1)}
\newcommand{\lastsymb}{\#}
\newcommand{\phisymb}{\phi_{\mathsf{symb}}}
\newcommand{\phicount}{\phi_{\mathsf{count}}}
\newcommand{\WF}{L_{\mathsf{symb}}}
\newcommand{\sep}{\$}

\newcommand{\sepvar}{\sigma}
\newcommand{\type}{\tau}
\newcommand{\lasttrans}{t_\text{final}}

\newcommand{\dist}{\,\!\!\!\!}
\newcommand{\extrans}[3]{\begin{array}[t]{c}#1\\|\\[-0.7ex]|\\\makebox[10pt][c]{\footnotesize$#2$}\\[-0.6ex]{\textcolor{red}{\scalebox{0.6}{#3}}}\end{array}}
\newcommand{\extsep}[2]{\begin{array}[b]{c}\textcolor{red}{\scalebox{0.6}{#2}}\\|\\[-0.6ex]|\\#1\end{array}}
\newcommand\Azero{A_\mathsf{zero}}
\newcommand\Aczero{A_{\overline{\mathsf{zero}}}}

\newcommand{\Atuples}{\mathit{Tuples}_A}
\newcommand{\Btuples}{\mathit{Tuples}_B}
\newcommand{\Atypes}{\mathit{Types}_A}
\newcommand{\Btypes}{\mathit{Types}_B}
\newcommand{\myimplies}{\mathrel{\text{implies}}}
\renewcommand{\qedsymbol}{\textup{q.e.d.}}
\renewcommand{\implies}{\mathrel{\Longrightarrow}}

\begin{document}

\maketitle

\begin{abstract}
We correct our proof of a theorem
stating that satisfiability of 
frequency linear-time temporal logic is undecidable [TASE 2012].
\end{abstract}


\section{Introduction}

In \cite{BolligDL12}, we introduced \emph{frequency linear-time temporal logic} (\fltl),
a quantitative extension of LTL over infinite words.
We stated that the logic has an undecidable satisfiability problem.
While the result is valid, the proof (a reduction from two-counter Minsky machines)
as described in \cite{BolligDL12} is flawed\footnote{This was
brought to our attention by Paul Gastin, whom we thank for his observation.}
 (cf.\ Remark~\ref{rem:main} in this note for an explanation).
Fortunately, it is fixable (cf.\ \cite{BouyerMM14} for a brief discussion), but
the corrected proof has not yet been presented.
This note rectifies this.
We take the opportunity to also unify and simplify our original reduction,
though the main idea remains essentially the same.

\paragraph{Related work.}

The logic \fltl is inspired by availability languages \cite{HoenickeMO10}.
It has recently been used for probabilistic systems
\cite{ForejtK15,ForejtKK15,PiribauerB20}. A related logic,
containing averaging modalities, was introduced
in \cite{BouyerMM14}. Recent work shows that already
very restricted frequency logics have an undecidable satisfiability
problem, e.g., the logic where the only temporal modalities are
\emph{future} and a \emph{majority past operator}~\cite{abs-2007-01233}.

\paragraph{Outline.}
In Section~\ref{sec:fLTL}, we recall the definition of \fltl and
state undecidability of its satisfiability problem.
In Section~\ref{sec:reduction}, we provide the reduction from the reachability problem in two-counter Minsky machines to the satisfiability problem of \fltl. 
Correctness of the reduction is shown in
Section~\ref{sec:correctness}.

\section{Frequency Linear-Time Temporal Logic}
\label{sec:fLTL}

\paragraph{Preliminaries.}
The set of natural numbers is $\N = \{0,1,2,\ldots\}$.
The cardinality of a finite set $S$ is denoted by $|S|$.

Let $\Sigma$ be an alphabet, i.e., a nonempty finite set.
The set of finite words over $\Sigma$ is denoted by $\Sigma^\ast$,
and the set of countably infinite words by $\Sigma^\omega$.
For a finite word $w = a_0 \ldots a_{k-1} \in \Sigma^\ast$ (where $a_i \in \Sigma$),
we denote by $|w|$ the length $k$ of $w$. In particular, $|\epsilon| = 0$
where $\epsilon$ denotes the empty word.
Given $A \subseteq \Sigma$ and $w = a_0 \ldots a_{k-1} \in \Sigma^\ast$,
we let $\occnumb{w}{A} = |\{i \in \{0,\ldots,k-1\}: a_i \in A\}|$
denote the number of occurrences of letters from $A$ in $w$.

For an infinite 
word $w = a_0a_1a_2\ldots \in \Sigma^\omega$
and $i \in \N$, we let
$\sufx{w}{i} = a_ia_{i+1}a_{i+2} \ldots$ denote the suffix of $w$
starting from position $i$.

\paragraph{Frequency Linear-Time Temporal Logic.}
The idea of \emph{frequency linear-time temporal logic} (\fltl) is to allow for relaxation of the until operator in terms of 
an annotated frequency.
The usual intuition for a formula $\phi \U \psi$ is that $\psi$ must hold at some
point in the future, and before that $\phi$ has to hold \emph{always}. 
Instead of ``always'', we consider the less strict formulation ``sufficiently often''
referring to a minimum ratio $r\in[0,1]$ of positions, the \emph{frequency} of 
$\phi$.

  The syntax of \fltl formulas over $\Sigma$ is 
  given by
  $$
    \varphi ::= a\ |\ \X\varphi\ |\ \varphi \myU{r} \varphi\ |\ \neg\varphi\ |\ \varphi \land \varphi
  $$
  where $a$ ranges over $\Sigma$ and each $\U$-operator is annotated by a rational number
  $r\in \Rat$ with $0\leq r \leq 1$. \fltl formulas are interpreted over words
  $w=a_0a_1a_2 \ldots\in\Sigma^\omega$, as follows:
  \[\begin{array}{lcl}
    w \models a        & \text{if} & a_0 = a\\[1ex]
    w \models \X\phi   & \text{if} & \sufx{w}{1} \models \phi\\[1ex]
    w \models \neg\phi & \text{if} & w\not\models \phi\\[1ex]
    w \models \varphi \myU{r} \psi    & \textrm{if} & 
        \exists j \in \N: \sufx{w}{j} \models \psi \textrm{ and } |\{i \in \{0,\ldots,j-1\}: \sufx{w}{i} \models \phi \}| \geq r \cdot j\\[1ex]
    w \models \varphi \land \psi & \text{if} & 
        w \models \varphi \textrm{ and } w \models \psi
  \end{array}\]

Consider, for example $a \myU{\frac12} b$ and the word $w =
cbaabbc^\omega$. The observed
frequency of $a$ until positions 1 and 5 is $0$ and $\frac25$, respectively, which is
too low. Yet, at position 4, $b$ is satisfied \emph{and} the frequency
constraint before that position is met. Thus, $w$ is a model since the frequency constraint
is not necessarily required to hold at the first position where $b$
holds.

We use $\phi \U \psi$ as an abbreviation for $\phi\myU{1} \psi$,
which corresponds to the classical until operator.
Moreover, we use standard abbreviations such as
$\true = \bigvee_{a \in \Sigma} a$,
$\false = \neg\true$,
$\phi \vee \psi = \neg(\neg \phi\wedge\neg\psi)$,
$(\phi \implies \psi) = \neg\phi \vee \psi$,
$\F \phi = \true\U \phi$, and $\G \phi = \neg \F \neg \phi$.
Finally, given a set $A \subseteq \Sigma$, we use $A$ as a shorthand for
$\bigvee_{a \in A} a$.

\medskip

In the remainder of this note, we will show the following theorem
(Theorem~2 from \cite{BolligDL12}):

\begin{theorem}
\label{thm:undecidability}
The following problem is undecidable:
Given an alphabet $\Sigma$ and an \fltl formula
$\phi$ over $\Sigma$, is there a word $w \in \Sigma^\omega$
such that $w \models \phi$\,?
\end{theorem}

\section{Reduction from Minsky Machines to \fltl Satisfiability}

\label{sec:reduction}

To prove the theorem, we proceed by reduction from
the (undecidable) reachability problem of two-counter Minsky
machines \cite{minsky_67} to the satisfiability problem
of \fltl.

\subsection{Minsky Machines}

A Minsky machine is equpped with two nonnegative counters that it can
increment, decrement, or test for zero using the set of \emph{operations}
$\IS = \{ \inc_1, \inc_2, \dec_1, \dec_2, \testz_1, \testz_2 \}$.
A \emph{counter configuration} is a pair 
$C = (\valone, \valtwo) \in \mathbb{N} \times\mathbb{N}$.
The effect of an operation on a counter configuration is described
by a relation ${\to} \subseteq (\N \times \N) \times \IS \times (\N \times \N)$.
We have $(\valone,\valtwo) \ttrans{\op} (m',n')$ if one of the following holds:
\begin{itemize}
\item $\op = \inc_1$ and $\valone' = \valone+1$ and $\valtwo' = \valtwo$,
\item $\op = \dec_1$ and $\valone' = \valone-1 \ge 0$ and $\valtwo' = \valtwo$,
\item $\op = \testz_1$ and $\valone' = \valone = 0$ and $\valtwo' = \valtwo$,
\item $\op = \inc_2$ and $\valone' = \valone$ and $\valtwo' = \valtwo+1$,
\item $\op = \dec_2$ and $\valone' = \valone$ and $\valtwo' = \valtwo-1 \ge 0$,
\item $\op = \testz_2$ and $\valone' = \valone$ and $\valtwo' = \valtwo = 0$.
\end{itemize}

Now, a \emph{(two-counter) Minsky machine} is a tuple
$\mathcal{M} = (\Loc, \Trans,\src,\trg,\opt,t_{\text{init}},t_\text{final})$
where $\Loc$ is a finite set of
\emph{locations} and $\Trans$ is a finite set of transitions. The mappings
$\src: \Trans \to \Loc$,
$\trg: \Trans \to \Loc$, and
$\opt: \Trans \to \IS$ assign to each transition
a source location, a target location, and an operation, respectively.
Finally, $t_\text{init}$ is the \emph{initial transition} and $t_\text{final}$
the \emph{final transition}. For technical reasons (and without
loss of generality), we assume that
$t_\text{init} \neq t_\text{final}$,
$\{\opt(t_\text{init}),\opt(t_\text{final})\} \cap \{\testz_1, \testz_2\} = \emptyset$,
and $\trg(t_\text{final}) \neq \src(t)$ for all $t \in \Trans$.

For, $\op \in \IS$, let $\Trans_\op = \{t \in \Trans : \opt(t) = \op\}$.
Moreover, we set $\Trans_\inc = \Trans_{\inc_1} \cup \Trans_{\inc_2}$,
$\Trans_\dec = \Trans_{\dec_1} \cup \Trans_{\dec_2}$, and
$\Trans_\testz = \Trans_{\testz_1} \cup \Trans_{\testz_2}$.

For two counter configurations $C = (\valone,\valtwo)$ and $C' = (\valone',\valtwo')$
and a transition $t \in \Trans$,
we let $C \ttrans{t} C'$ if $C \ttrans{\opt(t)} C'$
(abusing notation).

A \emph{computation} $\pi$ of $\mathcal{M}$ is a finite sequence
\begin{align}
\pi:~ C_0 \ttrans{t_1} C_1 \ttrans{t_2} C_{2} \ldots \ttrans{t_k} C_{k}
\label{lab:computation}
\end{align}
($k \ge 1$) of counter configurations $C_i = (\valone_i,\valtwo_i)$
and transitions $t_i$ such that $C_0 =
(0, 0)$, $t_1 = t_\text{init}$, and, for all $i \in \{2,\ldots,k\}$,
$\trg(t_{i-1}) = \src(t_i)$. The computation is \emph{successful}
if $t_k = t_\text{final}$ (note that this actually implies $k \ge 2$).

It is well-known that the following problem is undecidable \cite{minsky_67}:
Does a given Minky machine $\mathcal{M}$ have a successful
computation?

\subsection{From Minsky Machines to \fltl}

We now correct (and simplify) the construction of an \fltl formula from a given
Minsky machine that we presented in \cite{BolligDL12}.

We fix a Minsky machine $\mathcal{M} = (\Loc, \Trans,\src,\trg,\opt,t_{\text{init}},t_\text{final})$.
We will construct an \fltl formula over
the alphabet
\[\Sigma = \{\asymb,\bsymb,\hatasymb,\hatbsymb\} \cup \Trans \cup \{\sep_0,\sep_1,\sep_\testz\} \cup \{\#\}\]
that is satisfiable iff $\mathcal{M}$ has a successful computation.
This will prove Theorem~\ref{thm:undecidability}.

\begin{figure}[t]
\centering
\scalebox{0.9}{
$(0,0)
\ttrans{(\loc_0,\inc_1,\loc_1)} (1,0)
\ttrans{(\loc_1,\inc_1,\loc_2)} (2,0)
\ttrans{(\loc_2,\testz_2,\loc_3)} (2,0)
\ttrans{(\loc_3,\inc_2,\loc_4)} (2,1)
\ttrans{(\loc_4,\dec_1,\loc_5)} (1,1)
\ttrans{(\loc_5,\inc_2,\loc_6)} (1,2)
$
}
\\[1ex]
$\downarrow$
\\[1ex]
$\begin{array}{r}
\\[-0.5ex]
\overbrace{~~~~~~~~}^{\substack{\textup{counter}\\\textup{update}\\[0.5ex]x_1}}
~~~\;
\overbrace{~~~~~~~~~~~~~}^{\substack{\textup{counter}\\\textup{update}\\[0.5ex]x_2}}
~~\;
\overbrace{~~~~~~~~~~~~~~~}^{\substack{\textup{counter update}\\[0.5ex]x_3}}
~~~~~\;
\overbrace{~~~~~~~~~~~~~~~~\;}^{\substack{\textup{counter update}\\[0.5ex]x_4}}
~~~\;
\overbrace{\;~~~~~~~~~~~~~~~~}^{\substack{\textup{counter update}\\[0.5ex]x_5}}
~~~\;\;
\overbrace{~~~~~~~~~~~~~~~~\,~~}^{\substack{\textup{counter update}\\[0.5ex]x_6}}
~~~~~~~~\,\,
\\[-9.0ex]
\extsep{\sep_0}{0}\!\dist\extrans{t_{\textup{init}}}{(\loc_0,\inc_1,\loc_1)}{1} \hatasymb
\extsep{\sep_1}{3} \asymb \extrans{t_2}{(\loc_1,\inc_1,\loc_2)}{5} \hatasymb\, \hatasymb
\extsep{\sep_0}{8} \asymb\,\asymb \extrans{t_3}{(\loc_2,\testz_2,\loc_3)}{11} \hatasymb\, \hatasymb
\extsep{\sep_\testz}{14} \asymb\,\asymb \extrans{t_4}{(\loc_3,\inc_2,\loc_4)}{17} \hatasymb\, \hatasymb\, \hatbsymb
\extsep{\sep_1}{21} \asymb\,\asymb\,\bsymb
\extrans{t_5}{(\loc_4,\dec_1,\loc_5)}{25} \hatasymb\, \hatbsymb
\extsep{\sep_0}{28} \asymb\,\bsymb
\extrans{t_{\textup{final}}}{(\loc_5,\inc_2,\loc_6)}{31}
\hatasymb\, \hatbsymb\, \hatbsymb
\extsep{\sep_1}{35}\,\dist \#^\omega\\[-12ex]
\underbrace{\;~~~~~~~~~~~\;}_{\substack{y_1\\[0.5ex]\textup{carryover}}}
~~~
\underbrace{~~~~~~~~~~~~~~~~}_{\substack{y_2\\[0.5ex]\textup{carryover}}}
~~
\underbrace{~~~~~~~~~~~~~~~~~~~}_{\substack{y_3\\[0.5ex]\textup{carryover}}}
~~
\underbrace{\;~~~~~~~~~~~~~~~~~~~~}_{\substack{y_4\\[0.5ex]\textup{carryover}}}
~\;
\underbrace{~~~~~~~~~~~~~~~~\;\,}_{\substack{y_5\\[0.5ex]\textup{carryover}}}
~~~~~~~~~~~~~~~~~~~~~~\;
\\[8ex]
\end{array}$
\caption{Encoding of a successful computation\label{fig:encoding}}
\end{figure}

\paragraph{Idea of Encoding.}
Before we formalize the reduction, we illustrate it by means of an example.
Figure~\ref{fig:encoding} depicts a successful computation and its
encoding as an infinite word over $\Sigma$.

A configuration $(\valone,\valtwo) \in \N \times \N$ is encoded in unary as
$\asymb^\valone \bsymb^\valtwo$ or $\hatasymb^\valone \hatbsymb^\valtwo$.
White symbols denote a configuration before a transition, and
gray symbols a configuration after a transition. For instance,
the finite word $x_4$ (positions 15--20) represents the step $(2,0) \ttrans{t_4} (2,1)$ where
$t_4 = (\loc_3,\inc_2,\loc_4)$. After a transition, gray symbols
are ``transformed back'' into white symbols (carryover) so that
another transition can take place. We need these two kinds of symbols
(white and gray) to check whether counter updates are performed
correctly. The separator symbols $\sep_0,\sep_1,\sep_\testz$ will
help us in this task and also to check whether every carryover
is correct.

The syntactic form of the encoding will be checked by
a simple LTL formula $\phisymb$ without frequencies (it is, in a sense,
symbolic). The counting part is done by a formula $\phicount$
so that $\mathcal{M}$ has a successful computation iff
$\phisymb \wedge \phicount$ is satisfiable.
We formalize both formulas below.

Let us first give the idea how we ensure a correct carryover. To make sure that the carryover
in the word $y_5 \in \Sigma^\ast$ covering positions 25--30 is valid,
the following formula has to hold in the suffix starting from position 25:
\[
\displaystyle
\Psi = \bigwedge_{\substack{\type_1 \in \{\asymb,\hatasymb\}\\\type_2 \in \{\bsymb,\hatbsymb\}}}
\left(
\begin{array}{ll}
& (\type_1 \vee \type_2 \vee \Trans) \myU{\frac{1}{2}} t_{\textup{final}}\\[0.5ex]
\wedge\!\! & (\type_1 \vee \type_2 \vee \sep_0 \vee \sep_1 \vee \sep_\testz) \myU{\frac{1}{2}} t_{\textup{final}}
\end{array}
\right)
\]
In particular, at least half of the letters in $y_5$ have to be from
$\{\asymb,\bsymb\} \cup \Trans$ and at least
half of them from $\{\hatasymb,\hatbsymb\} \cup \{\sep_0,\sep_1,\sep_\testz\}$ (which are disjoint sets
so that both sets cover exactly one half of the positions).
Note that, in $y_5$,
there are exactly one transition ($t_5$) and exactly one separator symbol ($\sep_0$)
so that these two letters cancel each other out.
Therefore, there are as many letters from $\{\asymb,\bsymb\}$
as letters from $\{\hatasymb,\hatbsymb\}$, i.e.,
$\occnumb{y_5}{\{\asymb,\bsymb\}} = \occnumb{y_5}{\{\hatasymb,\hatbsymb\}}
 = \frac{|y_5|-2}{2} = 2$.
The same reasoning actually gives us
$\occnumb{y_5}{\{\asymb,\hatbsymb\}} = \occnumb{y_5}{\{\hatasymb,\bsymb\}}
 = \frac{|y_5|-2}{2} = 2$. Thus, 
$\occnumb{y_5}{\{\asymb,\bsymb\}} = \occnumb{y_5}{\{\hatasymb,\bsymb\}}$
and
$\occnumb{y_5}{\{\hatasymb,\hatbsymb\}}=
\occnumb{y_5}{\{\hatasymb,\bsymb\}} = \frac{|y_5|-2}{2} = 2$.
This implies that $\asymb$ and $\hatasymb$ occur equally often in $y_5$ and so do
$\bsymb$ and $\hatbsymb$. Now, formula $\Psi$ has to hold at every transition.
Using a backward argument allows us to show that the
equality required for the correct carryover
holds for $y_4y_5$, $y_3y_4y_5$, etc.\ and, thus, in $y_4$, $y_3$, and so on.

Now, let us explain how we ensure correctness of counter updates.
In word $x_6$, counter 2 is incremented.
That is, on the right of the transition symbol $t_{\textup{final}}$,
there should be an additional $\hatbsymb$-symbol.
To guarantee this, we consider that an $\inc_2$-transition and
$\hatbsymb$ are complementary.
In other words, an $\inc_2$-transition counts like a $\bsymb$-symbol.
Moreover, $\dec_2$ is like $\hatbsymb$, and similarly for counter~1.
Following the trick described above,
we use an until formula ensuring that, in $x_6$ (i.e., from position 29 \emph{until}
$\sep_1$ holds), counter updates are simulated correctly.
More precisely, our formula will make sure that
the symbols from $\{\asymb,\bsymb\} \cup \Trans_{\inc_2}$,
the symbols from $\{\asymb,\hatbsymb\}$,
the symbols from $\{\hatasymb,\bsymb\} \cup \Trans_{\inc_2}$, and
the symbols from $\{\hatasymb,\hatbsymb\}$
all cover at least half of $x_6$ (and, therefore, exactly half of it).
In particular, as $x_6$ includes one $\inc_2$-transition,
this implies that there is one more occurrence of
$\hatbsymb$ than of $\bsymb$.

It is important to note here why we actually have three separator symbols.
Letters from $\Trans_\testz$ and
$\sep_\testz$ are complementary, which guarantees that the counter
values remain unchanged by a zero transition.
Moreover, $\sep_0$ and $\sep_1$ are complementary and
alternate in the encoding (cf.\ Figure~\ref{fig:encoding}).
We employ formulas (roughly) of the form
$\sep_0 \implies \X(\ldots\, \myU{\frac{1}{2}} \mathit{last}_1)$ and
$\sep_1 \implies \X(\ldots\, \myU{\frac{1}{2}} \mathit{last}_0)$
where formula $\mathit{last}_1 = \sep_1 \mathrel{\wedge} \X\G \neg\sep_1$
refers to the very last $\sep_1$-symbol and accordingly for $\mathit{last}_0$.
This guarantees an equal number of occurrences of
the complementary symbols $\sep_0$ and $\sep_1$
in the infix that we want to check
so that $\sep_0$ and $\sep_1$ do not interfere with the counter values.
Moreover, $\sep_\testz$ is associated with a zero test and vice versa
(positions 11 and 14). Thus, there will be as many occurrences of $\sep_\testz$
as zero transitions.
Now, consider the $\sep_\testz$-labeled position 14 in the example.
The next occurrence of a symbol from $\{\sep_0,\sep_1\}$ is $\sep_1$.
By means of an until-formula that refers to the last $\sep_1$, as sketched above, we will guarantee 
a balanced number of separator symbols (and zero transitions)
in the infix whose counter values are supposed to be checked.
In fact, starting from $x_6$, a backward induction shows 
that all counter updates $x_6,x_5,\ldots,x_1$ are correctly simulated.

\medskip

Observe that we use two different notions of complementarity (actually,
two different partitions of $\Sigma \setminus \{\#\}$), one for verifying
carryovers and one for verifying counter updates.

\medskip

Let us formalize $\phisymb$ and $\phicount$.

\paragraph{The Formula $\phisymb$.}

The models of $\phisymb$ will be given by the set
$\WF \subseteq \Sigma^\omega$ of words of the form
(the $C_i$ indicate that certain subwords represent counter configurations)
\begin{align}
w = \sigma_0\;
\underbrace{\asymb^{\valone_0}\bsymb^{\valtwo_0}}_{C_0} \;t_1\;
\underbrace{\hatasymb^{\valone_1}\hatbsymb^{\valtwo_1}}_{C_1}\;
\sepvar_{1}\;
\underbrace{\asymb^{\valone_1'}\bsymb^{\valtwo_1'}}_{C_1'} \;t_2\;
\underbrace{{\hatasymb}^{\valone_2}{\hatbsymb}^{\valtwo_2}}_{C_2}\;
\sepvar_{2}\;
\ldots\;
\underbrace{\asymb^{\valone_{k-1}'}\bsymb^{\valtwo_{k-1}'}}_{C_{k-1}'} \;t_k\;
\underbrace{{\hatasymb}^{\valone_k}{\hatbsymb}^{\valtwo_k}}_{C_k}\;
\sepvar_{k}\;
\lastsymb^\omega
\label{lab:encoding}
\end{align}
(with $k \ge 1$) where
\begin{itemize}
\item[(P1)] $\valone_0 = \valtwo_0 = 0$, $t_1 = t_{\text{init}}$, and $t_k = t_{\text{final}}$,
\item[(P2)] for all $i \in \{2,\ldots,k\}$, $\trg(t_{i-1}) = \src(t_{i})$,
\item[(P3)] for all $i \in \{1,\ldots,k\}$, $\opt(t_i) = \testz_1$ implies $\valone_{i} = 0$ and
\item[(P4)] for all $i \in \{1,\ldots,k\}$, $\opt(t_i) = \testz_2$ implies $\valtwo_{i} = 0$,
\item[(P5)] $\sepvar_0,\ldots,\sepvar_k$ are as follows:
\begin{enumerate}
\item[(a)] $\sepvar_0 = \sep_{0}$ and, for all $i \in \{1,\ldots,k\}$ such that $\opt(t_i) \in \Trans_{\testz}$, we have \[\sepvar_i = \sep_\testz\,,\]
\item[(b)] let $1 \le i_1 <\ldots < i_p \le k$ be all the indices $i_j$ such that $\opt(t_{i_j}) \not\in \Trans_{\testz}$; then, for all $j \in \{1,\ldots,p\}$, \[\sepvar_{i_j} = \sep_{(j \mathop{\textup{mod}} 2)}\,.\]
\end{enumerate}
\end{itemize}
Note that, since $t_{\textup{init}} \neq t_{\textup{final}}$, we actually have $k \ge 2$.
With the computation $\pi$ from (\ref{lab:computation}), we associate
the word $w \in \Sigma^\omega$ from (\ref{lab:encoding})
where $\sepvar_0,\ldots,\sepvar_k$ are as required in (P5), and
$m_i' = m_i$ and $n_i' = n_i$ for all $i \in \{1,\ldots,k-1\}$.
We denote this encoding by $\enc{\pi} \in \Sigma^\omega$.
Note that the encoding also satisfies (P1)--(P4).

In the following, we use $\sep$ as a shorthand
for the formula $\sep_0 \vee \sep_1 \vee \sep_\testz$.
We first observe that there is indeed a standard (without frequencies) LTL formula
$\phisymb$ whose models are exactly the words from $\WF$.
The next formula ensures the form required in (\ref{lab:encoding}):
\[
\begin{array}{rl}
& \sep \wedge \X \bigl((\neg\#) \mathrel{\U} (\sep \wedge \X \G \#)\bigr)\\[1ex]
\wedge\!\!\! &
\G
\left(
\begin{array}{rl}
&  \sep \wedge \neg \X \#\\[1ex]
\implies\!\!\! &
\X (\asymb \mathrel{\U} (\bsymb \mathrel{\U} (\Trans \wedge \X (\hatasymb \mathrel{\U} (\hatbsymb \mathrel{\U} \sep)))))
\end{array}
\right)
\end{array}
\]
Moreover, the following formulas take care of conditions (P1)--(P5):
\begin{itemize}
\item[(P1)] $\X t_{\textup{init}} \wedge \F \bigl(t_{\textup{final}} \wedge \X(\G\neg \Trans)\bigr)$

\item[(P2)] $\G \bigwedge_{t \in \Trans}
\bigl(t ~\implies~ \neg \X \bigvee_{\substack{t' \in \Trans\\\trg(t) \neq \src(t')}}
(\neg\Trans)\U t'
 \bigr)$

\item[(P3)] $\displaystyle \G \bigl(\Trans_{\testz_1}
~\implies~ (\neg\,\hatasymb) \U \sep\bigr)$

\item[(P4)] $\displaystyle \G \bigl(\Trans_{\testz_2}
~\implies~ (\neg\,\hatbsymb) \U \sep\bigr)$

\item[(P5a)] $\sep_0 \wedge \G \bigl(\Trans_\testz ~\implies~ (\neg \sep) \U \sep_\testz \bigr)$

\item[(P5b)]
~\!\!\!\!\!$\begin{array}[t]{rl}
& \G \bigl((\Trans_\inc \cup \Trans_\dec) ~\implies~ (\neg \sep) \mathrel{\U} (\sep_0 \vee \sep_1) \bigr)\\[1ex]
\wedge\!\!\!& \G \bigl(\bigwedge_{\beta \in \{0,1\}} \bigl[(\sep_\beta \wedge \neg \X \#) ~\implies~ (\neg (\sep_0 \vee \sep_1)) \mathrel{\U} \sep_{1-\beta}\bigr]\bigr)
\end{array}$
\end{itemize}
We hence obtain $\phisymb$ as conjunction of those formulas.

\paragraph{The Formula $\phicount$.}

Towards $\phicount$,
the following notation will turn out to be useful:
\begin{alignat*}{3}
\Atuples &= \{\asymb,\hatasymb\} &&\times \{\bsymb,\hatbsymb\} && \times \{0,1\} \times \{\testz,\overline{\testz}\}\\
\Btuples &= \{\asymb,\hatasymb\} &&\times \{\bsymb,\hatbsymb\} &&\times \{0,1\}\\
\Atypes &= \{\asymb,\hatasymb\} &&\cup \{\bsymb,\hatbsymb\} && \cup \{0,1\} \cup \{\testz,\overline{\testz}\}\\
\Btypes &= \{\asymb,\hatasymb\} &&\cup \{\bsymb,\hatbsymb\} &&\cup \{0,1\}
\end{alignat*}
We define the following two partitions of $\Sigma \setminus \{\#\}$.
Partition $(A_{\type})_{_{\type \in \Atypes}}$ is used for checking counter updates,
and partition $(B_{\type})_{_{\type \in \Btypes}}$ for checking carryovers:
\begin{align*}
A_\asymb &= \{\asymb\} \cup \Trans_{\inc_1}
& 
B_\asymb &= \{\asymb\}
\\
A_{\hatasymb} &= \{\hatasymb\} \cup \Trans_{\dec_1}
&
B_{\hatasymb} &= \{\hatasymb\}
\\
A_\bsymb &= \{\bsymb\} \cup \Trans_{\inc_2}
&
B_\bsymb &= \{\bsymb\}
\\
A_{\hatbsymb} &= \{\hatbsymb\} \cup \Trans_{\dec_2}
&
B_{\hatbsymb} &= \{\hatbsymb\}
\\
A_0 &= \{\sep_0\} \qquad \Azero = \Trans_{\testz}
&
B_0 &= \Trans
\\
A_1 &= \{\sep_1\} \qquad \Aczero = \{\sep_\testz\}
&
B_1 &= \{\sep_0,\sep_1,\sep_\testz\}
\end{align*}

Recall that we consider some letters to be complementary, such as those from $A_\asymb$ and those from $A_\hatasymb$.
More precisely, for $\type \in \{\asymb,\hatasymb\} \cup \{\bsymb,\hatbsymb\} \cup \{0,1\} \cup \{\testz,\overline{\testz}\}$,
we let $\overline{\type}$ denote the ``complement'' of $\type$, i.e.,
$\overline{\asymb} = \hatasymb$, $\overline{\hatasymb} = \asymb$,
$\overline{\bsymb} = \hatbsymb$, $\overline{\hatbsymb} = \bsymb$,
$\overline{0} = 1$, $\overline{1} = 0$, and
$\overline{\overline{\testz}} = \testz$.

From the following easy lemma (whose proof can be found in Section~\ref{sec:model-comp}), we will derive frequency until formulas that guarantee that
certain complementary letters will occur equally often in a word.

\begin{lemma}\label{lem:equality}
Let $w \in (\Sigma \setminus \{\#\})^\ast$.
\begin{itemize}
\item[(a)] Suppose that, for all $(\type_1,\type_2,\type_3,\type_4) \in \Atuples$, we have
\[\occnumb{w}{A_{\type_1}} + \occnumb{w}{{A_{\type_2}}} + \occnumb{w}{{A_{\type_3}}} + \occnumb{w}{{A_{\type_4}}} \ge \frac{|w|}{2}\,.\]
Then,
$\occnumb{w}{A_{\type}} = \occnumb{w}{{A_{\overline{\type}}}}$ for all $\type \in \Atypes$.
\item[(b)] Suppose that, for all $(\type_1,\type_2,\type_3) \in \Btuples$, we have
\[\occnumb{w}{B_{\type_1}} + \occnumb{w}{{B_{\type_2}}} + \occnumb{w}{{B_{\type_3}}} \ge \frac{|w|}{2}\,.\]
Then, $\occnumb{w}{B_{\type}} = \occnumb{w}{{B_{\overline{\type}}}}$ for all $\type \in \Btypes$.
\end{itemize}
\end{lemma}

For $\beta \in \{0,1\}$, we now define auxiliary formulas
\begin{align*}
\mathit{last}_{\beta} &= \sep_\beta \wedge \X\G \neg \sep_\beta\\
\mathit{next}_{\beta} &= \X\, \bigl( (\neg (\sep_0 \vee \sep_1)) \U \sep_\beta \bigr)\,.
\end{align*}
The first formula holds exactly at the last $\sep_\beta$-position.
The second formula holds exactly at those positions such that the very
next symbol from $\{\sep_0,\sep_1\}$ 
(exists and) 
is $\sep_\beta$.

The next formulas are justified by Lemma~\ref{lem:equality}.
For $\beta \in \{0,1\}$, the formula
\begin{align*}
\Phi_\beta &= \X
\displaystyle
\bigwedge_{
\substack{
(\type_1,\type_2,\type_3,\type_4) \\
\in \Atuples
}}
\bigl((A_{\type_1} \vee A_{\type_2} \vee A_{\type_3} \vee A_{\type_4}) \myU{\frac{1}{2}} \mathit{last}_{\beta}\bigr)\\
\intertext{will take care of counter updates. To take into account carryover, we define}
\Psi &=
\displaystyle
\bigwedge_{
\substack{
(\type_1,\type_2,\type_3) \\
\in \Btuples
}}
\bigl((B_{\type_1} \vee B_{\type_2} \vee B_{\type_3}) \myU{\frac{1}{2}} t_{\textup{final}}\bigr)\,.
\end{align*}

With this, we finally let
\[
\phicount = \G
\left(
\begin{array}{rlcl}
& (\sep \wedge \mathit{next}_0 &\implies& \Phi_0)\\[0.5ex]
\wedge\!\!& (\sep \wedge \mathit{next}_1 &\implies& \Phi_1)\\[0.5ex]
\wedge \!\!& (\Trans \wedge \neg\lasttrans &\implies& \Psi)
\end{array}
\right)\,.
\]

\begin{remark}\label{rem:main}
Note that the until formulas used in $\Phi_0$, $\Phi_1$, and $\Psi$
all refer to unique positions
in terms of formulas $\mathit{last}_{\beta}$ and $t_{\textup{final}}$, respectively.
If satisfied, one can show, using a backward argument, that the very next
$\sep_\beta$-position and, respectively, $\Trans$-position are also possible reference points.
This latter property was not ensured in \cite{BolligDL12}, resulting in the flaw.
\end{remark}

We will show the following lemma, which establishes undecidability of the satisfiability
problem for \fltl:
\begin{lemma}\label{lem:correctness}
The following statements are equivalent:
\begin{itemize}
\item The Minsky machine $\mathcal{M}$ has a successful computation.
\item The \fltl formula $\phisymb \wedge \phicount$ is satisfiable.
\end{itemize}
\end{lemma}

\section{Correctness of the Reduction (Proof of Lemma~\ref{lem:correctness})}
\label{sec:correctness}

In this section, we prove Lemma~\ref{lem:correctness}.

\subsection{Every Computation is a Model}

The first direction is easy:

\begin{lemma}\label{lem:comp-model}
Let $\pi$ be a successful computation of $\mathcal{M}$.
Then, $\enc{\pi} \models \phisymb \wedge \phicount$.
\end{lemma}

\begin{proof}
Suppose $\pi$ is like in equation~(\ref{lab:computation}) and
let $w = \enc{\pi}$ as given by equation (\ref{lab:encoding}).

First, note that $w \models \phisymb$.
Indeed, (P5) holds by definition of $\enc{\pi}$, and
(P1)--(P4) hold by the definition of a successful computation.

\medskip

Next, we check $w \models \phicount$.
Let $i \in \{1,\ldots,k-1\}$ and
\begin{align*}
w_i =
\underbrace{
t_i\;
\underbrace{\hatasymb^{\valone_i}\hatbsymb^{\valtwo_i}}_{C_i}\;
\sepvar_{i}\;
\underbrace{\asymb^{\valone_i}\bsymb^{\valtwo_i}}_{C_i}
\;
\ldots\;
\;t_{k-1}\;
\underbrace{{\hatasymb}^{\valone_{k-1}}{\hatbsymb}^{\valtwo_{k-1}}}_{C_{k-1}}\;
\sepvar_{k-1}
\underbrace{\asymb^{\valone_{k-1}}\bsymb^{\valtwo_{k-1}}}_{C_{k-1}}
}_{=: z_i}
\;t_\textup{final}\;
\underbrace{{\hatasymb}^{\valone_k}{\hatbsymb}^{\valtwo_k}}_{C_k}\;
\sepvar_{k}\;
\lastsymb^\omega
\end{align*}
be the suffix of $w$ starting at $t_i$.
We show that $w_i \models \Psi$.
For $(\type_1,\type_2,\type_3) \in \Btuples$, we have
\[
\occnumb{z_i}{B_{\type_1}} = \sum_{j = i}^{k-1} \valone_j \text{~~~~~and~~~~~}
\occnumb{z_i}{B_{\type_2}} = \sum_{j = i}^{k-1} \valtwo_j \text{~~~~~and~~~~~}
\occnumb{z_i}{B_{\type_3}} = {k-i}.
\]
Thus,
\[\occnumb{z_i}{B_{\type_1}} + \occnumb{z_i}{B_{\type_2}} + \occnumb{z_i}{B_{\type_3}} = \frac{|z_i|}{2}\,.\]
This implies $w_i \models \Psi$.

Now, let $\beta \in \{0,1\}$, $i,j \in \{1,\ldots,k\}$, and
\[\begin{array}{l}
w_i = \sigma_{i-1}\;
\underbrace{
\underbrace{
\underbrace{u_{i-1}}_{C_{i-1}} \;t_{i}\;
\underbrace{v_{i}}_{C_{i}}\;
\ldots\;
}_{=: z_i}
\sepvar_{j}\;
\underbrace{u_{j}}_{C_{j}} \;t_{j+1}\;
\underbrace{v_{j+1}}_{C_{j+1}}\;
\ldots\;
\sepvar_{k-1}\;
\underbrace{u_{k-1}}_{C_{k-1}} \;t_k\;
\underbrace{v_k}_{C_k}\;
}_{=: z_i'}
\sepvar_{k}\;
\lastsymb^\omega\,.
\end{array}\]
such that 
$\sepvar_{j} \in A_{1-\beta}$
and $\sepvar_k \in A_\beta$ are
the last occurrences of
$\sep_{1-\beta}$ and $\sep_\beta$, respectively.
Note that $\occnumb{z_i}{A_\testz} = \occnumb{z_i}{\Aczero}$
and $\occnumb{z_i'}{A_\testz} = \occnumb{z_i'}{\Aczero}$
since zero-test transitions are coupled with $\sep_\testz$ and vice versa (P5).
Again due to (P5), labels from $\{\sep_0,\sep_1\}$ alternate. Thus,
\begin{alignat*}{3}
&w_i \models \mathit{next}_{1-\beta} &~\myimplies~& \occnumb{z_i}{A_0} = \occnumb{z_i}{A_1}\\
&w_i \models \mathit{next}_\beta &~\myimplies~& \occnumb{z_i'}{A_0} = \occnumb{z_i'}{A_1}\,.
&\end{alignat*}
Finally, for all $p \in \{i,\ldots,k\}$ and $x_p = u_{p-1} t_{p} v_{p}$,
we have
$\occnumb{x_p}{A_\asymb} = \occnumb{x_p}{A_\hatasymb}$ and
$\occnumb{x_p}{A_\bsymb} = \occnumb{x_p}{A_\hatbsymb}$
due to the definition of the transition relation.
Thus, we also obtain
$\occnumb{z_i}{A_\asymb} = \occnumb{z_i}{A_\hatasymb}$,
$\occnumb{z_i}{A_\bsymb} = \occnumb{z_i}{A_\hatbsymb}$,
$\occnumb{z_i'}{A_\asymb} = \occnumb{z_i'}{A_\hatasymb}$, and
$\occnumb{z_i'}{A_\bsymb} = \occnumb{z_i'}{A_\hatbsymb}$.

Thus, if we have $w_i \models \mathit{next}_{1-\beta}$, then
$\sum_{p \in \{1,2,3,4\}} \occnumb{z_i}{A_{\type_p}} = \frac{|z_i|}{2}$
for all $(\type_1,\type_2,\type_3,\type_4) \in \Atypes$.
Moreover, if $w_i \models \mathit{next}_{\beta}$, then
$\sum_{p \in \{1,2,3,4\}} \occnumb{z_i'}{A_{\type_p}} = \frac{|z_i'|}{2}$
for all $(\type_1,\type_2,\type_3,\type_4) \in \Atypes$.
This shows that $w \models \phicount$.
\end{proof}

\subsection{Every Model is a Computation}\label{sec:model-comp}

Before we prove the other, more interesting direction, we provide
the proof of Lemma~\ref{lem:equality}:

\begin{proof}[Proof of Lemma~\ref{lem:equality}]
We only show (a). Part (b) is along exactly the same lines.
Suppose $\occnumb{w}{A_{\type_1}} + \occnumb{w}{{A_{\type_2}}} + \occnumb{w}{{A_{\type_3}}} + \occnumb{w}{{A_{\type_4}}} \ge \frac{|w|}{2}$
for all $(\type_1,\type_2,\type_3,\type_4) \in \Atuples$.
As $(A_\type)_{\type \in \Atypes}$ is a partition of $\Sigma \setminus \{\#\}$, we also have
$\sum_{p \in \{1,2,3,4\}} \occnumb{w}{A_{\type_p}} + \occnumb{w}{A_{\overline{\type_p}}} = |w|$
for all $(\type_1,\type_2,\type_3,\type_4) \in \Atypes$.
Altogether, we get
$\occnumb{w}{A_{\type_1}} + \occnumb{w}{{A_{\type_2}}} + \occnumb{w}{{A_{\type_3}}} + \occnumb{w}{{A_{\type_4}}} = \frac{|w|}{2}$
for all tuples $(\type_1,\type_2,\type_3,\type_4) \in \Atuples$.
In particular,
\begin{align*}
  &\occnumb{w}{A_{\asymb}} + \occnumb{w}{{A_{\bsymb}}} + \occnumb{w}{{A_{0}}} + \occnumb{w}{{A_{\testz}}}\\
=~ &\occnumb{w}{A_{\hatasymb}} + \occnumb{w}{{A_{\bsymb}}} + \occnumb{w}{{A_{0}}} + \occnumb{w}{{A_{\testz}}}\\
=~ &\occnumb{w}{A_{\hatasymb}} + \occnumb{w}{{A_{\hatbsymb}}} + \occnumb{w}{{A_{0}}} + \occnumb{w}{{A_{\testz}}}\\
=~ &\occnumb{w}{A_{\hatasymb}} + \occnumb{w}{{A_{\hatbsymb}}} + \occnumb{w}{{A_{1}}} + \occnumb{w}{{A_{\testz}}}\\
=~ &\occnumb{w}{A_{\hatasymb}} + \occnumb{w}{{A_{\hatbsymb}}} + \occnumb{w}{{A_{1}}} + \occnumb{w}{{\Aczero}}
\end{align*}
Thus, $\occnumb{w}{A_{\type}} = \occnumb{w}{{A_{\overline{\type}}}}$ for all $\type \in \Atypes$.
\end{proof}

Now, we are ready to prove:

\begin{lemma}\label{lem:model-comp}
Let $w \in \Sigma^\omega$ be a model of $\phisymb \wedge \phicount$.
Then, $\mathcal{M}$ has a successful computation.
\end{lemma}

\begin{proof}
Let $w \in \Sigma^\omega$ be a model of $\phisymb \wedge \phicount$.
Since $w \models \phisymb$, we have $w \in \WF$. Thus, it is of the
form as in equation (\ref{lab:encoding}) such that (P1)--(P5) are satisfied:
\[\begin{array}{r}
\overbrace{~~~~~~~~~~~~~~~\;}^{x_1}
\hspace{1.8em}
\overbrace{~~~~~~~~~~~~~~~\;}^{x_2}
\hspace{10.6em}
\overbrace{~~~~~~~~~~~~~~~~~~~\;\,}^{x_{k-1}}
\hspace{2.7em}
\overbrace{~~~~~~~~~~~~~~~~\,}^{x_{k}}
\hspace{3.5em}
\\[-0.7ex]
\sepvar_0\;
\underbrace{u_0}_{C_0}
\underbrace{
\;t_1\;
\underbrace{v_1}_{C_1}\;
\sepvar_{1}\;
\underbrace{u_1}_{C_1'}}_{y_1}
\;t_2\;
\underbrace{v_2}_{C_2}
\;
\sepvar_{2}\;
\ldots\;
\underbrace{
\;t_{k-2}\;
\underbrace{v_{k-2}}_{C_{k-2}}
\sepvar_{k-2}\;
\underbrace{u_{k-2}}_{C_{k-2}'}}_{y_{k-2}}
\underbrace{
\;t_{k-1}\;
\underbrace{v_{k-1}}_{C_{k-1}}
\;
\sepvar_{k-1}\;
\underbrace{u_{k-1}}_{C_{k-1}'}}_{y_{k-1}}
\;t_k\;
\underbrace{v_k}_{C_k}\;
\sepvar_{k}\;
\lastsymb^\omega
\end{array}\]
where
$u_0 = \asymb^{\valone_0}\bsymb^{\valtwo_0}$,
$v_i = \hatasymb^{\valone_i}\hatbsymb^{\valtwo_i}$ for $i \in \{1,\ldots,k\}$, and
$u_i = \asymb^{\valone_i'}\bsymb^{\valtwo_i'}$ for $i \in \{1,\ldots,k-1\}$.
As indicated above, we define words $x_i = u_{i-1}t_iv_i$ (for $i \in \{1,\ldots,k\}$)
and $y_i = t_iv_i\sepvar_iu_i$ (for $i \in \{1,\ldots,k-1\}$) as well as
counter configurations $C_i = (\valone_i,\valtwo_i)$
and $C_i' = (\valone_i',\valtwo_i')$ of $\mathcal{M}$.
Moreover, we let $y_k = \epsilon$.

By (P2), we have $\trg(t_{i-1}) = \src(t_{i})$ for all $i \in \{2,\ldots,k\}$.
We will show that,
for all $i \in \{1,\ldots,k\}$, we have $C_{i-1} \ttrans{t_i} C_{i} = C_{i}'$.
In particular, we still have to check that the counter values are updated according
to $\opt(t_i)$.

We first prove the following claim:

\begin{claim}\label{claim:equal-x}
For all $i \in \{1,\ldots,k\}$,
$\occnumb{x_i}{A_\asymb} = \occnumb{x_i}{A_\hatasymb}$ and $\occnumb{x_i}{A_\bsymb} = \occnumb{x_i}{A_\hatbsymb}$.
\end{claim}

\begin{proof}[Proof of Claim~\ref{claim:equal-x}]
\renewcommand{\qedsymbol}{$\Diamond$}
Recall that $t_k \not\in \Trans_\testz$ (by convention)
so that $\sepvar_k \in \{\sep_0,\sep_1\}$.
Let $\beta \in \{0,1\}$, $i,j \in \{1,\ldots,k\}$, and
\[\begin{array}{l}
~~~~~~~~~~~~~~
\overbrace{~~~~~~~~~~~~~~\;}^{x_i}
~~~~~~~~~~~~~~~~~~~~~~~~~~~~~~~~~~~~~~~~~~~
\overbrace{~~~~~~~~~~~~~~~\;}^{x_k}
\\[-1ex]
w_i = \sigma_{i-1}\;
\underbrace{
\underbrace{
\underbrace{u_{i-1}}_{C_{i-1}} \;t_{i}\;
\underbrace{v_{i}}_{C_{i}}\;
\ldots\;
}_{=: z_i}
\sepvar_{j}\;
\underbrace{u_{j}}_{C_{j}} \;t_{j+1}\;
\underbrace{v_{j}}_{C_{j+1}}\;
\ldots\;
\sepvar_{k-1}\;
\underbrace{u_{k-1}}_{C_{k-1}} \;t_k\;
\underbrace{v_k}_{C_k}\;
}_{=: z_i'}
\sepvar_{k}\;
\lastsymb^\omega\,.
\end{array}\]
such that
$\sepvar_{j} \in A_{1-\beta}$
and $\sepvar_k \in A_\beta$ are
the last occurrences of
$\sep_{1-\beta}$ and $\sep_\beta$, respectively.
We proceed by induction, showing the statement for
$x_k$, then for $x_{k-1}$ and so on.
In the base case $i = k$, we have $j < i$ and therefore consider $z_i = \epsilon$.
We distinguish two (very similar) cases:
\begin{description}
\item[Case $w_i \models \mathit{next}_{1-\beta}$:]
Since $w_i \models \Phi_{1-\beta}$ and by Lemma~\ref{lem:equality}(a), we have $\occnumb{z_i}{A_{\asymb}} = \occnumb{z_i}{A_{\hatasymb}}$ and
$\occnumb{z_i}{A_{\bsymb}} = \occnumb{z_i}{A_{\hatbsymb}}$.
For $\tau \in \{\asymb,\hatasymb,\bsymb,\hatbsymb\}$, let
$D_\tau = \sum_{p =i+1}^{j}\occnumb{x_p}{A_{\type}}$.
By induction hypothesis (or when $i \ge j$),
$D_\asymb = D_\hatasymb$ and $D_\bsymb = D_\hatbsymb$.
We have
\begin{align*}
\occnumb{x_i}{A_{\asymb}} &= \occnumb{z_i}{A_{\asymb}} - D_\asymb\\
\occnumb{x_i}{A_{\hatasymb}} &= \occnumb{z_i}{A_{\hatasymb}} - D_\hatasymb\\
\occnumb{x_i}{A_{\bsymb}} &= \occnumb{z_i}{A_{\bsymb}} - D_\bsymb\\
\occnumb{x_i}{A_{\hatbsymb}} &= \occnumb{z_i}{A_{\hatbsymb}} - D_\hatbsymb\,,
\end{align*}
which implies
$\occnumb{x_i}{A_\asymb} = \occnumb{x_i}{A_\hatasymb}$ and
$\occnumb{x_i}{A_\bsymb} = \occnumb{x_i}{A_\hatbsymb}$.

\item[Case $w_i \models \mathit{next}_{\beta}$:]
This case is almost identical.
We only replace $1-\beta$ with $\beta$, $z_i$ with $z_i'$, and $j$ with $k$.
\end{description}

Thus, in both cases, we indeed obtain $\occnumb{x_i}{A_\asymb} = \occnumb{x_i}{A_\hatasymb}$ and
$\occnumb{x_i}{A_\bsymb} = \occnumb{x_i}{A_\hatbsymb}$.
\end{proof}

We now show the corresponding claim for carryovers:

\begin{claim}\label{claim:equal-y}
For all $i \in \{1,\ldots,k-1\}$,
$\occnumb{y_i}{\{\asymb\}} = \occnumb{y_i}{\{\hatasymb\}}$ and
$\occnumb{y_i}{\{\bsymb\}} = \occnumb{y_i}{\{\hatbsymb\}}$.
\end{claim}

\begin{proof}[Proof of Claim~\ref{claim:equal-y}]
\renewcommand{\qedsymbol}{$\Diamond$}
Let $i \in \{1,\ldots,k-1\}$. We have
\[
\underbrace{
y_{i}
\underbrace{
y_{i+1} \ldots y_{k-1}
}_{=: y}
}_{=: y'}
 \,t_{\textup{final}}\, v_k \sepvar_k \#^\omega ~\models
\displaystyle
\bigwedge_{
\substack{
(\type_1,\type_2,\type_3) \\
\in \Btuples
}}
\bigl((B_{\type_1} \vee B_{\type_2} \vee B_{\type_3}) \myU{\frac{1}{2}} t_{\textup{final}}\bigr)\,.\]
By Lemma~\ref{lem:equality}(b), we obtain
$\occnumb{y'}{\{\asymb\}} = \occnumb{y'}{\{\hatasymb\}}$ and
$\occnumb{y'}{\{\bsymb\}} = \occnumb{y'}{\{\hatbsymb\}}$.
Since we also have
$y_{i+1} \ldots y_{k-1}
 \,t_{\textup{final}}\, v_k \sepvar_k \#^\omega \models
\Psi$
and by Lemma~\ref{lem:equality}(b), we get
$\occnumb{y}{\{\asymb\}} = \occnumb{y}{\{\hatasymb\}}$ and
$\occnumb{y}{\{\bsymb\}} = \occnumb{y}{\{\hatbsymb\}}$.
Further, note that $\occnumb{y_i}{B_\type} = \occnumb{y'}{B_\type} - \occnumb{y}{B_\type}$
for all $\type \in \Btypes$.
We obtain $\occnumb{y_i}{\{\asymb\}} = \occnumb{y_i}{\{\hatasymb\}}$ and
$\occnumb{y_i}{\{\bsymb\}} = \occnumb{y_i}{\{\hatbsymb\}}$ as desired.
\end{proof}

By Claim~\ref{claim:equal-y}, we have $C_i = C_i'$ for all $i \in \{1,\ldots,k-1\}$.
Thus, the carryover is performed correctly.
Using Claim~\ref{claim:equal-x}, we can now take care of the counter updates and show:
\begin{claim}\label{claim:counter}
For all $i \in \{1,\ldots,k\}$, we have $(\valone_{i-1},\valtwo_{i-1}) \ttrans{\opt(t_i)} (\valone_{i},\valtwo_{i})$.
\end{claim}

\begin{proof}[Proof of Claim~\ref{claim:counter}]
\renewcommand{\qedsymbol}{$\Diamond$}
Suppose $i \in \{1,\ldots,k\}$.
We distinguish several cases:
\begin{description}
\item[Case $\opt(t_{i}) = \inc_1$:]
We have $t_{i} \in A_\asymb$.
By Claims~\ref{claim:equal-x} and \ref{claim:equal-y},
we get
$\valone_{i} = \valone_{i-1} + 1$ (as required for the increment of counter~1)
as well as $\valtwo_{i} = \valtwo_{i-1}$ (the second counter value remains unchanged). 

\item[Case $\opt(t_{i}) = \inc_2$:]
We have $t_{i} \in A_\bsymb$.
By Claims~\ref{claim:equal-x} and \ref{claim:equal-y},
$\valone_{i} = \valone_{i-1}$ and $\valtwo_{i} = \valtwo_{i-1} + 1$.

\item[Case $\opt(t_{i}) = \dec_1$:]
We have $t_{i} \in A_\hatasymb$.
By Claims~\ref{claim:equal-x} and \ref{claim:equal-y},
$\valone_{i} = \valone_{i-1} - 1$ and
$\valtwo_{i} = \valtwo_{i-1}$.

\item[Case $\opt(t_{i}) = \dec_2$:]
We have $t_{i} \in A_\hatbsymb$.
By Claims~\ref{claim:equal-x} and \ref{claim:equal-y},
$\valone_{i} = \valone_{i-1}$ and
$\valtwo_{i} = \valtwo_{i-1} - 1$.

\item[Case $\opt(t_{i}) = \testz_1$:]
We have $t_{i} \in A_\testz$.
By Claims~\ref{claim:equal-x} and \ref{claim:equal-y},
$\valone_{i} = \valone_{i-1}$ and
$\valtwo_{i} = \valtwo_{i-1}$. Moreover, by (P3), $\valone_{i} = 0$.

\item[Case $\opt(t_{i}) = \testz_2$:]
We have $t_{i} \in A_\testz$.
By Claims~\ref{claim:equal-x} and \ref{claim:equal-y},
$\valone_{i} = \valone_{i-1}$ and
$\valtwo_{i} = \valtwo_{i-1}$. Moreover, by (P4), $\valtwo_{i} = 0$.
\qedhere
\end{description}
\end{proof}

Altogether, we have shown that,
for all $i \in \{1,\ldots,k\}$, we have $C_{i-1} \ttrans{t_i} C_{i}$.
This concludes the proof of Lemma~\ref{lem:model-comp}.
\end{proof}

\bibliographystyle{plain}
\bibliography{main}

\end{document}